\documentclass[11pt]{article}
\usepackage[margin=1.0in]{geometry}
\usepackage{mathtools}
\usepackage{amssymb,amsthm}
\usepackage{graphicx}
\graphicspath{{figures/}}
\usepackage{color}
\usepackage{lineno}
\usepackage{color} 
\usepackage{setspace}

\mathtoolsset{showonlyrefs,showmanualtags}

\usepackage{chemfig}

\newcommand{\jao}[1] {}
\renewcommand{\jao}[1] {{\color{red}{\textbf{JAO:} #1}}}
\newcommand{\jwb}[1] {}
\renewcommand{\jwb}[1] {{\color{blue}{\textbf{JWB:} #1}}}
\newcommand{\pt}[1] {}
\renewcommand{\pt}[1] {{\color{orange}{\textbf{PT:} #1}}}

\newcommand{\R}{\mathbb{R}}
\newcommand{\lap}{\mbox{$\cal L$}}
\newcommand{\ra}{\rightarrow}

\newtheorem{claim}{Claim}

\usepackage{letltxmacro}
\LetLtxMacro{\originaleqref}{\eqref}
\renewcommand{\eqref}{Eq.\,\originaleqref}

\spacing{1}

\makeatletter
\renewcommand{\maketitle}{\bgroup\setlength{\parindent}{0pt}
\begin{flushleft}
\textbf{\@title}
\@author
\end{flushleft}\egroup
}
\makeatother

\title{\Large Thermodynamic bounds on ultrasensitivity in covalent switching \\[2pt]}

\author{
Jeremy A. Owen$^{1,\ast,\dagger}$, 
Pranay Talla$^{2,\ast,\ddagger}$, 
John W. Biddle$^{3,\S}$, 
Jeremy Gunawardena$^{3,\vert\vert}$}

\begin{document}
\maketitle

{\small
$^{1}$Department of Physics, Massachusetts Institute of Technology, Cambridge, MA, USA

$^{2}$Horace Greeley High School, Chappaqua, NY, USA

$^{3}$Department of Systems Biology, Harvard Medical School, Boston, MA, USA \\[1pt]

$^{\ast}$These authors contributed equally \\[1pt]

$^{\dagger}$Current address: Department of Chemistry, Princeton University, NJ 08540, USA \\[1pt]
$^{\ddagger}$Current address: Columbia College, Columbia University, New York, NY 10027, USA \\[1pt]
$^{\S}$Current address: Holy Cross College, Notre Dame, IN 46556, USA \\[1pt]

$^{\vert\vert}$Corresponding author: Jeremy Gunawardena (\tt{jeremy@hms.harvard.edu})
}

\vspace{0.1in}
Dated: \today

\vspace{0.6in}
{\noindent
ABSTRACT \\[1em]
\textbf{Switch-like motifs are among the basic building blocks of biochemical networks. A common motif that can serve as an ultrasensitive switch consists of two enzymes acting antagonistically on a substrate, one making and the other removing a covalent modification. To work as a switch, such covalent modification cycles must be held out of thermodynamic equilibrium by continuous expenditure of energy. Here, we exploit the linear framework for timescale separation to establish tight bounds on the performance of any covalent-modification switch, in terms of the chemical potential difference driving the cycle. The bounds apply to arbitrary enzyme mechanisms, not just Michaelis-Menten, with arbitrary rate constants, and thereby reflect fundamental physical constraints on covalent switching.}

\pagebreak

\renewcommand\linenumberfont{\normalfont\small}
\setlength{\parindent}{15pt}
\doublespacing



\section*{Introduction}

The covalent modification cycle is a ubiquitous motif in biochemical networks. In this motif, a forward modifying enzyme, $E$, covalently attaches a modifying group to a substrate, $S$, thereby converting it from an unmodified state, $S_0$, to a modified state, $S_1$; and a reverse demodifying enzyme, $F$, removes the modifying group, converting $S_1$ back to $S_0$ (Figure \ref{fig:phospho}(a)). Phosphorylation is one the best known forms of covalent modification but many others are known \cite{psg12, Ree2018, walsh} and new forms of modification continue to be uncovered \cite{fmm19}. The substrate, $S$, may be a protein, in which case modifications are referred to as post-translational modifications, but they may also occur on small molecules. For the modification cycles considered here, the modifying group is a small chemical moiety obtained from a donor, such as a phosphate group obtained from ATP. Polypeptide modifications, such as ubiquitin, require a more complex cascade of enzymes for covalent attachment to their substrates and fall outside the scope of this paper. 

The antagonistic structure of covalent modification cycles was difficult to understand at first---why simultaneously attach a modifying group and also remove it?---and led to them being referred to in the older literature as ``futile cycles'' \cite{smu68}. In fact, the forward and reverse enzymes allow the balance of $S_1$ and $S_0$, measured, for instance, by the ratio of their steady-state concentrations, $[S_1]/[S_0]$, to be maintained away from the value it would have at chemical equilibrium. In other words, a covalent modification cycle can act as a switch (Figure \ref{fig:phospho}(b)), in which the value of $[S_1]/[S_0]$ is modulated by changing the levels of the forward or reverse enzymes \cite{Goldbeter_Koshland_1981, Stadtman_Chock_1977}. The idea of a biochemical switch becomes more natural in the context of cellular information processing and such switches have been found to play key roles in signal transduction \cite{bblhwk, Chen2001, Torrecilla2006,Ventura2010}, gene regulation \cite{Berger2001,Lehembre2000}, the cell cycle \cite{Fisher2012, Okamura2010}, and metabolism \cite{dcg12, Kustu1984}.

The operation of a covalent modification cycle relies on the continued presence of donor molecules to provide modifying groups. The cycle is driven by the chemical potential difference between the donor molecules, such as ATP, and the corresponding molecular residues after modification and demodification, such as ADP and inorganic phosphate P$_\mathrm{i}$. This chemical potential difference is maintained by core metabolic processes within the cell. It is akin to a battery in an electronic circuit and a modification cycle thereby operates away from thermodynamic equilibrium through continuous dissipation of energy \cite{Goldbeter_Koshland_1987,Qian_2003}. 

\begin{figure}
	\begin{center}
		\includegraphics[scale=0.6]{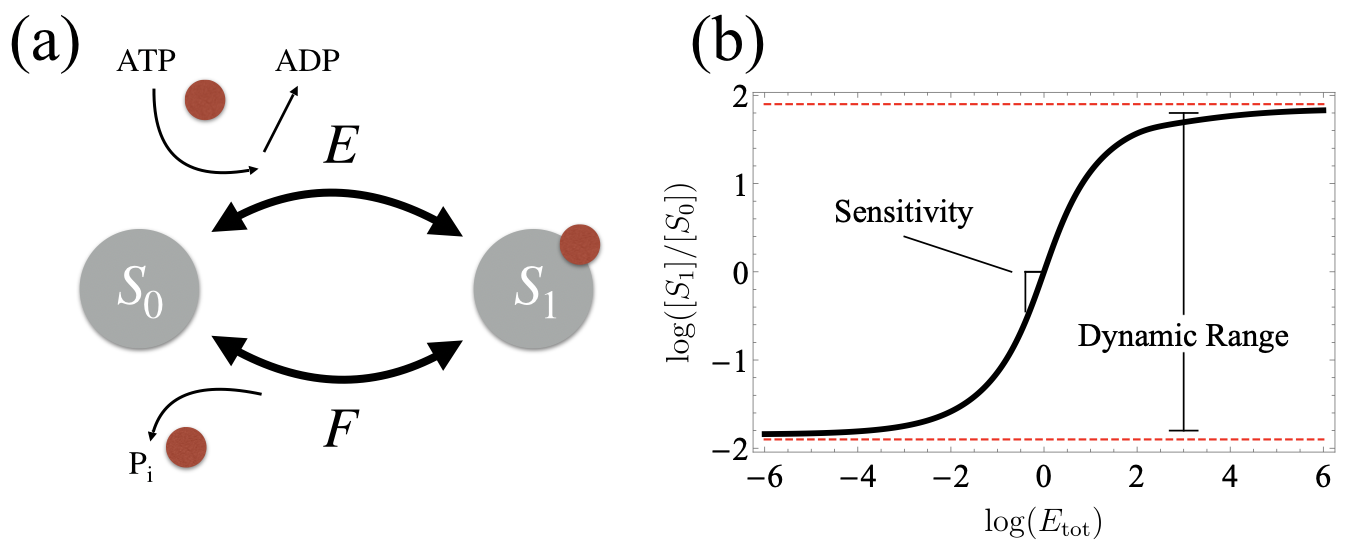}
	\end{center}
	\caption{{\bf Covalent modification cycle.} (a) Schematic of a phosphorylation-dephosphorylation cycle in which kinase $E$ modifies the substrate $S$ by covalent addition of a phosphate group (brown disc), donated by ATP, and phosphatase $F$ removes the modification by hydrolysis to release inorganic phosphate, $P_\mathrm{i}$. $S_0$ and $S_1$ denote the unphosphorylated and phosphorylated forms of $S$. (b) Covalent modification gives rise to a switch-like relationship between the total amounts of enzymes, here of the kinase $E$, and the steady-state substrate concentrations, which we quantify as illustrated by the \emph{sensitivity} and \emph{dynamic range}.}
	\label{fig:phospho}
\end{figure}

In seminal work, Goldbeter and Koshland performed a mathematical analysis of covalent modification cycles under the assumption of Michaelis-Menten kinetics for the modifying enzymes---catalysis proceeds via a single intermediate enzyme--substrate complex and product formation is irreversible \cite{Goldbeter_Koshland_1981}. The chemical reactions are:
\begin{equation}
\begin{array}{c}
E + S_0 \rightleftharpoons ES \to E + S_1 \\
F + S_1 \rightleftharpoons FS \to F + S_0
\end{array}
\label{e-mm}
\end{equation}

\noindent where $ES$ and $FS$ are enzymatic intermediates. When substrate is relatively abundant so that both enzymes are saturated, they found that this system can exhibit unlimited sensitivity to the concentrations of the modifying enzymes. For example, the logarithmic sensitivity, 
\begin{equation}
  \frac{\partial \log ([S_1]/[S_0])}{\partial \log E_\mathrm{tot}}
\label{e-plogs}
\end{equation}

\noindent which we refer to hereafter as the \emph{sensitivity}, can be made as large as desired, by varying rate constants. By way of comparison, if the relationship between $[S_1]/[S_0]$ and $E_\mathrm{tot}$ were described by a Hill function, $x^H/(1 + x^H)$, then the maximum sensitivity would be the Hill coefficient $H$. A sensitivity greater than $1$ is said to be ``ultrasensitive'' \cite{Ferrell2014}. 

The unlimited sensitivity found by Goldbeter and Koshland is physiologically implausible. It arises from the unrealistic assumption of irreversibility in the Michaelis-Menten reaction mechanism in \eqref{e-mm}. Although such an assumption has been nearly universal in quantitative studies of biochemical networks, it would have surprised Michaelis and Menten \cite{gun-mm, mm13} and its dangers have been repeatedly pointed out \cite{bblhwk,nbg20, oawmc}. In their in-vitro studies, Michaelis and Menten measured initial reaction rates, when product was not present, so that irreversibility was a reasonable assumption. But, since then, \eqref{e-mm} has been widely used in contexts, such as modification cycles, in which product is very much present and rebinding of product to enzyme is to be expected. Its continuing use has sometimes been justified on the grounds that modification and demodification reactions are often physiologically irreversible, in the sense that product is rarely converted back into substrate. However, enzyme mechanisms typically involve greater complexity than the simple Michaelis-Menten mechanism, with multiple intermediates and routes \cite{fersht}, and they may be physiologically irreversible despite product rebinding. The graph-theoretic linear framework for timescale separation \cite{gun-mt} allows realistic general reaction mechanisms to be analyzed in which such distinctions can be made. In previous work, we have analyzed modification cycles with realistic enzyme mechanisms and derived formulas for their switching capability in the limit of high substrate \cite{dcg12, xg11b}. The switching sensitivity is no longer unbounded but is now limited by the parameters of the switch.

A further difficulty with the irreversibility of \eqref{e-mm} is that it implies infinite entropy production. In reality, every reaction is reversible, although physiological conditions may make the reverse rate extremely low. As we will see below, it is the ratio of forward to reverse rates which yields the finite rate of entropy production. The Michaelis-Menten mechanism is therefore unsuitable for a thermodynamic analysis. Noting this, Qian \cite{Qian_2003} studied a minimal elaboration of the Goldbeter-Koshland cycle, in which both enzymes follow a fully reversible Michaelis-Menten mechanism. This mechanism still proceeds via a single intermediate complex but product can rebind. Qian found the relationship between the chemical potential difference driving the cycle and figures of merit, such as the sensitivity, of a switch based on the system.  

In the linear framework, realistic enzyme mechanisms can be analyzed with reversibility assumed throughout. We use this approach here to establish bounds on the switching \textit{dynamic range} and \textit{sensitivity} (Figure \ref{fig:phospho}(b)) of any covalent modification cycle in which the forward and reverse enzyme each follow their own realistic enzyme mechanism. We also show explicitly that these thermodynamic bounds can be approached as closely as desired.  Our work generalizes the analysis of Goldbeter and Koshland \cite{Goldbeter_Koshland_1981,Goldbeter_Koshland_1987} and the subsequent work of Qian \cite{Qian_2003, Qian_2007}, and reveals fundamental physical constraints, free from restrictive enzymological assumptions. 

\section*{Results}

\subsection*{Covalent modification cycles}

For present purposes, a covalent modification cycle is any chemical reaction network with mass action kinetics, built out of any number of reactions of the form \cite{Thomson_Gunawardena_2009, xg11b}:
\begin{equation}
\begin{split}
E + S_0 \rightleftharpoons (ES)_i \\ 
E + S_1   \rightleftharpoons (ES)_i \\ 
(ES)_i \rightleftharpoons (ES)_j \\
F + S_0  \rightleftharpoons (FS)_i \\ 
F + S_1 \rightleftharpoons (FS)_i \\ 
(FS)_i \rightleftharpoons (FS)_j.
\end{split}\label{gkloopcrn}
\end{equation}

Any such network can be viewed as a detailed realization of the schematic modification cycle illustrated in Figure \ref{fig:phospho}(a). This class of models encompasses the irreversible cycle studied by Goldbeter and Koshland, but also much more biochemically realistic ones reflecting complex enzymology, with any number of intermediates, and, in particular, the reversibility of enzymes.

The reactions \eqref{gkloopcrn} imply the conservation of the total concentration of substrate $S_\mathrm{tot}$: 
\begin{equation}
S_\mathrm{tot} = [S_0] + [S_1] +\sum_{i}[(ES)_i] + \sum_{j}[(FS)_j],  \label{subconserve}
\end{equation}
as well as the total concentrations of both enzymes, $E_\mathrm{tot}$ and $F_\mathrm{tot}$:
\begin{align}
E_\mathrm{tot} &= [E] + \sum_{i}[(ES)_i] \notag \\
F_\mathrm{tot} &= [F] + \sum_{j}[(FS)_j]. \label{enzconserve}
\end{align}

The assumption of mass-action kinetics gives a system of polynomial differential equations for the time evolution of the concentration of each chemical species. The equations are arrived at by summing the individual contributions to the rate of formation/destruction of each species due to each reaction.

Given any fixed choice of rate constants and the conserved quantities $S_\mathrm{tot}$, $E_\mathrm{tot}$, $F_\mathrm{tot}$, a covalent modification cycle admits a \textit{unique} steady state or dynamical fixed point \cite{dcg12}. This fact allows us to view steady-state quantities, such as the steady-state value of the ratio $[S_1]/[S_0]$, as \textit{functions} of the conserved quantities. For a general covalent modification cycle, the polynomial equations satisfied at steady state, implicitly defining this functional relationship, can be very complicated---having arbitrarily many terms and rate constants appearing in them. Nevertheless, they possess a basic structure, set by the schema \eqref{gkloopcrn}, that will enable us to apply a powerful algebraic approach---the linear framework---to make general statements about them.

\subsection*{Background on the linear framework}

\begin{figure}
	\begin{center}
		\includegraphics[scale=0.55]{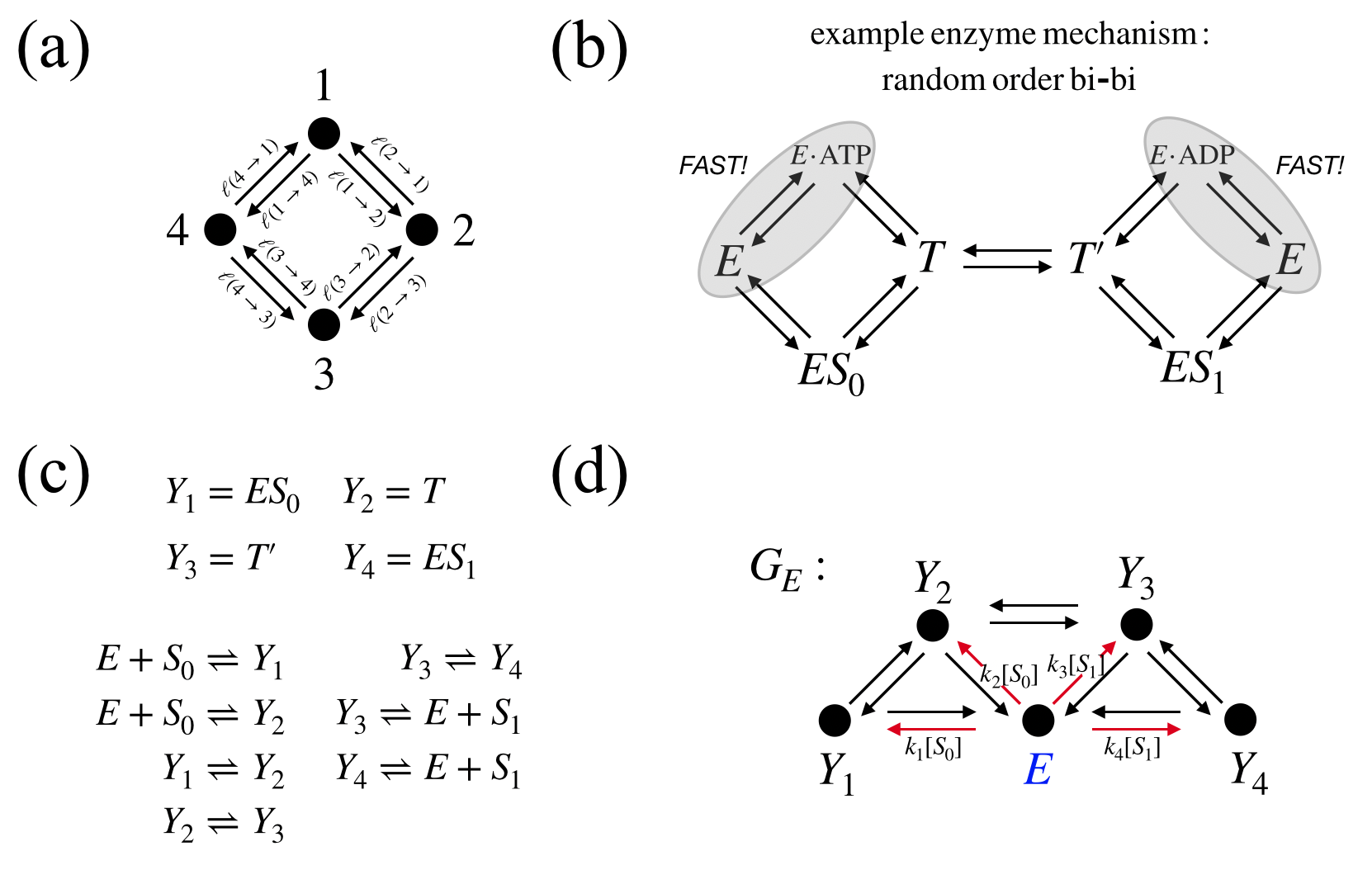}
	\end{center}
	\caption{{\bf Modeling enzyme mechanisms using the linear framework.}  (a) An example of a linear framework graph. (b) An example of a realistic enzyme mechanism, in which two substrates ($\mathrm{ATP}$ and $S_0$) bind to an enzyme $E$ in any order to form a ternary complex $T$, which is transformed into a complex $T'$ from which two products ($\mathrm{ADP}$ and $S_1$) are released in any order. To be able to model this mechanism using the grammar \eqref{e-grm}, the binding/unbinding of at least one of the substrates and one of the products to the bare enzyme must come to a rapid equilibrium (denoted by the gray ovals).  (c) Subject to this timescale assumption, the network in (b) can be cast in terms of the grammar \eqref{e-grm}. (d) A linear framework graph $G_E$ corresponding to the mechanism shown in (c).}
	\label{fig:graphs}
\end{figure}

Here, we briefly introduce the linear framework. In the following subsection, we will specialize to the case of covalent modification cycles. The framework was introduced in \cite{gun-mt}. A recent review \cite{kmg22} discusses the material needed here and should be consulted for more details and background. The framework revolves around finite, simple, directed graphs with labeled edges; an example graph is shown in Figure \ref{fig:graphs}(a). For the application considered here, the graph vertices, denoted by $1, 2, \cdots, n$, represent chemical species, the edges, denoted $i \ra j$, represent reactions and the labels, denoted $\ell(i \ra j)$, represent positive rates with dimensions of (time)$^{-1}$. The labels may be algebraic expressions that include time-varying concentrations of chemical species. We will discuss such labels further in the next section; for the present, to explain the machinery of the framework, we regard the labels as constant symbols. 

Let $G$ be a linear framework graph. Such a graph can be naturally given dynamics by assuming that each edge is a chemical reaction with the corresponding label as the rate constant for mass-action kinetics. Since an edge has only a single source vertex, the dynamics must be linear and is therefore described by a matrix differential equation, 
\begin{equation}
\frac{dx}{dt} = \lap(G).x \,,
\label{e-lapd}
\end{equation}
where $x = (x_1, \cdots, x_n)^T \in \R^n$ is the time-dependent column vector of vertex concentrations (here, $^T$ denotes transpose) and the linear operator, $\lap(G)$, is the Laplacian matrix of $G$. Since material is only moved around the graph, without being created or destroyed, there is a conservation law, at all times, $x_1 + \cdots + x_n = x_\mathrm{tot}$; equivalently, each column of $\lap(G)$ sums to zero.

We will be concerned with steady states, $x^*$, of systems described by graphs, so that $(dx/dt)|_{x = x^*} = 0$. Accordingly, $x^* \in \ker\lap(G)$. It can be shown that, if $G$ is strongly connected, then $\dim\ker\lap(G) = 1$. Recall that a graph is strongly connected if any two distinct vertices, $i$ and $j$, can be connected by a directed path, $i = i_1 \ra i_2 \ra \cdots \ra i_k = j$. A canonical basis element $\rho(G) \in \ker\lap(G)$ may be determined in terms of the edge labels by using the Matrix-Tree Theorem (MTT) of graph theory \cite{Hill_1966, King_Altman_1956, Mirzaev_Gunawardena_2013, Tutte_1948}. If $H$ is any subgraph of $G$, let $w(H)$ denote the product of the edge labels over the edges in $H$,
\[ w(H) = \prod_{i \ra j \in H} \ell(i \ra j) \,. \]
Recall that a spanning tree is a connected subgraph of $G$ that contains each vertex of $G$ (spanning) and has no cycles if edge directions are ignored (tree). It is said to be rooted at $i$ if $i$ is the only vertex with no outgoing edge (which orients the tree). Let $\Theta_i(G)$ denote the set of spanning trees of $G$ that are rooted at $i$.  Then, the MTT shows that,
\begin{equation}
\rho(G) = \sum_{T_i \in \Theta_i(G)} w(T_i) \,.
\label{e-MTT}
\end{equation}
Since $x^*$ must be proportional to $\rho(G) \in \ker\lap(G)$, it is straightforward to obtain the ratio of steady states in terms only of the edge labels,
\begin{equation}
\frac{x^*_i}{x^*_j} = \frac{\sum_{T_i \in \Theta_i(G)} w(T_i)}{\sum_{T_j \in \Theta_j(G)} w(T_j)} \,,
\label{e-main}
\end{equation}
and we will exploit this below.

\subsection*{Modeling a covalent modification cycle}

In previous work, post-translational modification systems, like the covalent modification cycle of Figure \ref{fig:phospho}(a), were modeled as interacting systems of linear framework graphs \cite{dcg12, xg11b}. This approach is also reviewed in \cite{kmg22}, which may be consulted for more details. An important feature of this approach is that enzyme reaction mechanisms can be substantially more general than the conventional Michaelis-Menten mechanism in \eqref{e-mm}, allowing in particular for reversibility and multiple intermediate complexes and thereby addressing the problems described in the Introduction. Specifically, an enzyme mechanism may be composed of any appropriate reactions from the ``grammar'',
\begin{equation}
E + S_* \ra Y_* \hspace{2em} Y_* \ra Y_* \hspace{2em} Y_* \ra E + S_* \,,
\label{e-grm}
\end{equation}
which transform between substrate and product. Here, we have used ``$_{*}$'' as a generic subscript to avoid index proliferation. $S_*$ denotes a substrate form, such as $S_0$ or $S_1$ in \eqref{e-mm}, and $Y_*$ denotes an intermediate complex, such as $E S$ in \eqref{e-mm}. The only requirement we place on the mechanism is that there must be no ``dead-end'' intermediate complexes that can be formed but not destroyed. We will see that this amounts to supposing that a certain corresponding graph, described below, is strongly connected.

The Michaelis-Menten mechanism in \eqref{e-mm}, or the reversible version used by Qian in \cite{Qian_2003}, can be constructed from the grammar in \eqref{e-grm}. But the grammar can also capture more complicated mechanisms, such as the  ``random order bi-bi'' mechanism \cite{seg93} pictured in Figure \ref{fig:graphs}(b), in which an enzyme has two substrates that are bound in either order and forms two products that are released in either order. Such mechanisms are important for forward modifying enzymes, like kinases, which use a secondary substrate, like ATP, to donate the modifying group (Figure \ref{fig:phospho}(a)) and release a secondary product, like ADP. Similarly, the enzyme mechanisms implied by the covalent modification cycle in \eqref{gkloopcrn} can also be accommodated in \eqref{e-grm}. 

In order to cast realistic enzyme mechanisms in our grammar, we must make an approximation, because in \eqref{e-grm}, substrates are not permitted to bind to intermediate complexes, so secondary substrates cannot be explicitly represented. Instead, as explained in more detail in \cite{xg11b}, we will assume that secondary-substrate binding occurs as a ``rapid equilibrium'' with either the free enzyme $E$ or some intermediate complex $Y_*$, so that the concentration of the secondary substrate can be absorbed into the appropriate rate constant. We additionally assume that the concentration of this secondary substrate is held constant---a reasonable assumption in the case of many forward modifying enzymes, especially kinases, where ATP concentration is held constant by background metabolic processes despite fluctuations in demand. Subject to these assumptions, the mechanism in Figure \ref{fig:graphs}(b) can be well-approximated by the reaction network shown in Figure \ref{fig:graphs}(c), which is a realization of the grammar in \eqref{e-grm}.

Any mechanism for an enzyme $E$ expressed in the grammar may be represented by a linear framework graph $G_E$ (Figure \ref{fig:graphs}(d)) in which the vertices are the free enzyme $E$ and the intermediate complexes $Y_*$, the edges correspond to the reactions in the mechanism and the labels correspond to the reaction rates. The time-dependent concentrations of the substrate forms $S_*$ appear in the labels for those edges in which a substrate form binds to the free enzyme. Accordingly, as can be seen in the example in Figure \ref{fig:graphs}(d), only those edges outgoing from the vertex $E$, which are colored red, have edge labels depending on the substrate concentrations. In keeping with the reversibility of the enzyme mechanism in Figure \ref{fig:graphs}(c), every transition in the graph $G_E$ in Figure 2(d) is reversible. We will always require that $G_E$ be strongly connected, and in particular that no intermediates form irreversibly.

The linear Laplacian dynamics on $G_E$, as given by \eqref{e-lapd}, is merely a rewriting of the dynamics of the enzyme mechanism under mass-action kinetics. But this construction will allow us to apply the MTT (i.e. in the form \eqref{e-main}) to algebraically express the variables associated to the vertices in terms of the edge labels.

The general form of the covalent modification cycle in \eqref{gkloopcrn} can easily be constructed within the grammar in \eqref{e-grm}: it consists of two enzymes $E$ and $F$, with arbitrary mechanisms obeying the grammar \eqref{e-grm}, and two inter-converting substrate forms $S_0$ and $S_1$. The reverse modifying enzyme, $F$, may follow a different mechanism to that of $E$. This is typically the case in reality, as the removal of a modifying group is often a single-substrate hydrolysis reaction. The ability to realistically represent the difference between the mechanisms of $E$ and $F$ is an important benefit of our approach. 

Applying the linear framework to a covalent modification cycle, we get two graphs: $G_E$, whose vertices are $E$ and its intermediates, and $G_F$, whose vertices are $F$ and its intermediates. Since $[S_0]$ and $[S_1]$ appear in the labels of the edges of $G_E$ and $G_F$ directed out of $E$ and $F$, respectively, they will also appear on the right hand side of \eqref{e-main}, but they do so in a limited way. Specifically, since each rooted spanning tree has at most one edge directed out of any vertex, they appear only linearly. This observation leads to the following simple relations at steady-state \cite{xg11b}: 
\begin{equation}
\begin{array}{rcl}
\sum_{i}\frac{[(ES)_i]}{[E]} & = & \frac{[S_0]}{K_0^E} + \frac{[S_1]}{K_1^E} \\
\sum_{j}\frac{[(FS)_j]}{[F]} & = & \frac{[S_0]}{K_0^F} + \frac{[S_1]}{K_1^F} 
\end{array}
\label{EIgraph}
\end{equation}
where $K_0^E, K_1^E, K_0^F, K_1^F$ are the \emph{total generalized Michaelis-Menten constants} (tgMMCs) of the covalent modification cycle \cite{Thomson_Gunawardena_2009, xg11b}. The tgMMCs depend only on rate constants and have dimensions of concentration.

Substituting expressions like \eqref{EIgraph} into the original polynomial steady-state equations yields a similarly compact expression for the substrate ratio (Eq. (13), \cite{xg11b}): 
\begin{equation}
\frac{[S_1]}{[S_0]} = \frac{c_0^E [E]+ c_0^F[F]}{c_1^E [E] + c_1^F [F]}
\label{SIgraph}
\end{equation}
where the quantities $c_0^E, c_1^E, c_0^F, c_1^F$ are the \emph{total generalized catalytic efficiencies} (tgCEs) of the covalent modification cycle \cite{Thomson_Gunawardena_2009, xg11b}. These quantities depend only on rate constants, and have dimensions of $\left(\text{concentration} \times \text{time}\right)^{-1}$. 

The three conservation laws and the relations \eqref{EIgraph} and \eqref{SIgraph} can now be combined to yield two equations involving only $[S_0]$ and $[S_1]$ as variables:
\begin{equation}
S_\mathrm{tot} = [S_0] + [S_1] + E_\mathrm{tot} \left(\frac{[S_0]/K_0^E+[S_1]/K_1^E}{1+[S_0]/K_0^E+[S_1]/K_1^E}\right) + F_\mathrm{tot} \left(\frac{[S_0]/K_0^F+[S_1]/K_1^F}{1+[S_0]/K_0^F+[S_1]/K_1^F}\right),\label{Invariant-1}
\end{equation}
\begin{equation}
\left(\frac{E_\mathrm{tot}}{F_\mathrm{tot}}\right)\left(1+\frac{[S_1]}{K_1^F}+\frac{[S_0]}{K_0^F}\right)\left(c_0^E[S_0]-c_1^E[S_1]\right)=\left(1+\frac{[S_0]}{K_0^E}+\frac{[S_1]}{K_1^E}\right)\left(c_1^F[S_1]-c_0^F[S_0]\right).\label{Invariant-2}
\end{equation}
For the purposes of understanding the steady-state dependence of $[S_0]$ and $[S_1]$ on the conserved quantities, all the possible complexity permitted by the schema in \eqref{gkloopcrn}, and all the freedom to choose rate constants, has been reduced to eight generalized parameters---the four tgMMCs and the four tgCEs.

\subsection*{Thermodynamic constraints on the tgCEs}

\begin{figure}
	\begin{center}
		\includegraphics[scale=0.55]{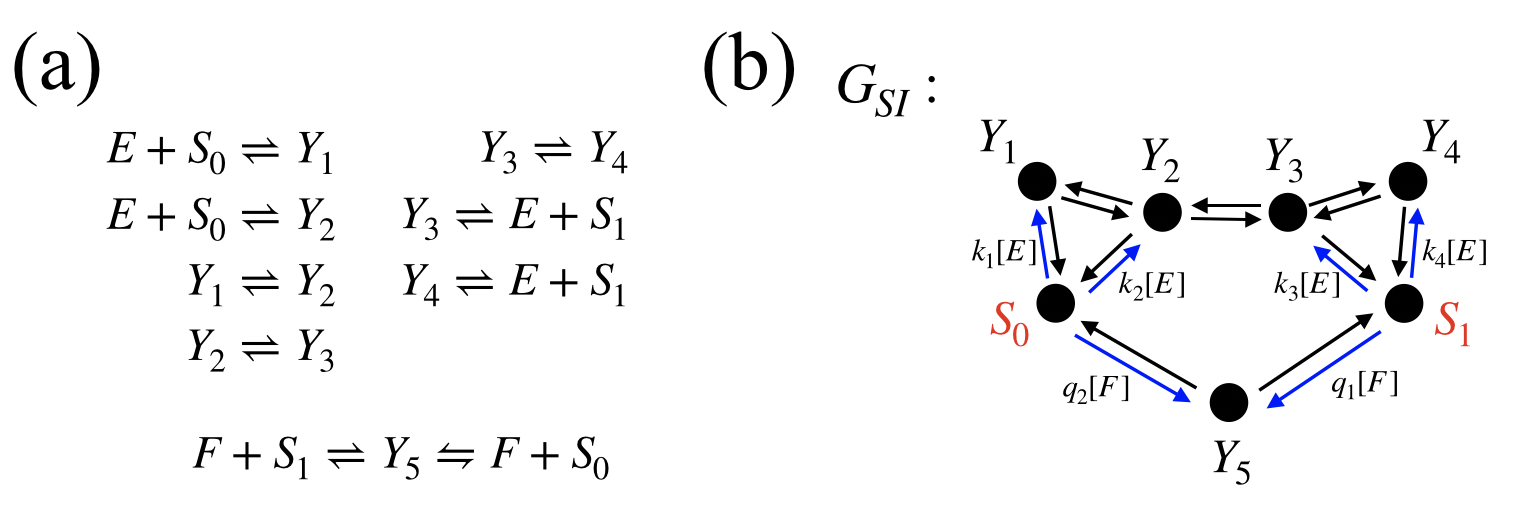}
	\end{center}
	\caption{{\bf An example covalent modification cycle and the graph $G_{SI}$.} (a) A covalent modification cycle where the enzyme $E$ obeys the random order bi-bi mechanism from Figure \ref{fig:graphs} and the enzyme $F$ acts by a reversible Michaelis-Menten mechanism. (b) The associated linear framework graph $G_{SI}$ involving the substrates and the intermediates.}
	\label{fig:gSI}
\end{figure}

Now we turn to thermodynamics.  We will assume from now on that every graph is reversible, in the sense that, if $i \ra j$, then also $j \ra i$. It is also important that the reverse edge $j \ra i$ represents the process that is the \textit{time-reverse} of that represented by $i \ra j$ and not just some other process for returning from $j$ to $i$ \cite{acs20}.

At thermodynamic equilibrium, each reaction occurs with the same frequency as its time-reverse. This principle of detailed balance \cite{lewis1925new} implies a relation between the \textit{rate constants} of reactions that form a cycle.  Consider, for example, an arbitrary cycle of reactions in a covalent modification cycle taking $S_0$ to $S_1$ via the enzyme $E$ and back to $S_0$ via the the enzyme $F$:
\begin{equation}
S_0   \underset{k_{-1}}{\stackrel{k_1 [E]}{\rightleftharpoons}}  (ES)_1  \underset{k_{-2}}{\stackrel{k_2}{\rightleftharpoons}}  \cdots  \rightleftharpoons (ES)_n   \underset{k_{-n}[E]}{\stackrel{k_n }{\rightleftharpoons}}  S_1   \underset{q_{-1}}{\stackrel{q_{1} [F]}{\rightleftharpoons}} (FS)_1  \rightleftharpoons \cdots  \underset{q_{-2}}{\stackrel{q_{2}}{\rightleftharpoons}}  (FS)_m   \underset{q_{-m} [F]}{\stackrel{q_{m}}{\rightleftharpoons}}  S_0  .
\end{equation}
For this cycle, detailed balance at thermodynamic equilibrium implies that
\begin{equation}
\frac{k_1 k_2 \cdots q_2 q_m}{k_{-1} k_{-2}  \cdots q_{-2} q_{-m}} = 1.
\label{db}
\end{equation}

Recall that some of the rate constants $k_i, q_i$ appearing above may conceal concentrations of cofactors we do not explicitly model. Holding these concentrations at fixed values away from their equilibrium values drives the system out of equilibrium into a nonequilibrium steady state. The log ratio of rates on the left hand side of \eqref{db}---known as the \textit{thermodynamic force}, or cycle affinity---quantifies the breaking of detailed balance and can often be identified as the entropy produced in the environment when the cycle is traversed, in units of the Boltzmann constant, $k_B$, \cite{esposito2010three, schnakenberg1976network, sei11}. For a chemical system held out of equilibrium by the presence of species assumed to have fixed concentration (either because there is a very large reservoir of them or because a fixed concentration is actively maintained), this will depend on a chemical potential difference $\Delta \mu$. For example, suppose $E$ is a kinase and $F$ a phosphatase, so that one complete realization of the cycle entails the binding of one molecule of ATP, the release of one molecule of ADP, and the release of one molecule of P$_{i}$. Then, 
\begin{equation}
\log\left(\frac{k_1 k_2 \cdots q_2 q_m}{k_{-1} k_{-2}  \cdots q_{-2} q_{-m}}\right) = \frac{\mu_\mathrm{ATP} - \mu_\mathrm{ADP} - \mu_\mathrm{P_\mathrm{i}}}{k_\mathrm{B}T} = \frac{\Delta \mu}{k_\mathrm{B}T} \approx 20\mbox{ to }30,
\end{equation}
under typical physiological conditions. In this work, we focus on the natural case in which the thermodynamic force is the \textit{same} around any cycle in which $E$ makes the modification of $S$ and $F$ removes it. In that case, we can show that: 
\begin{equation}
\log\left(\frac{c_0^E c_1^F}{c_1^E c_0^F}\right) = \frac{\Delta \mu}{k_\mathrm{B}T}. \label{Force-from-CE}
\end{equation}
In the general case, where the force is different around different cycles, the left hand side is bounded by the largest force. Without loss of generality, we take $c_0^F/c_1^F < c_0^E/c_1^E$ so that $\Delta \mu > 0$.

\eqref{Force-from-CE} is the physical constraint on the parameters of a general covalent modification cycle from which all our results descend. To prove it, it will be convenient to consider another kind of linear framework graph, denoted $G_{SI}$, whose vertices are the substrate forms and the intermediate complexes. $G_{SI}$ is an amalgam of $G_E$ and $G_F$ together with the substrate forms $S_0$ and $S_1$ and it is more convenient for a thermodynamic analysis of the cycle.\footnote{$G_{SI}$ is distinct from the graph on the substrate forms introduced in prior work involving the linear framework \cite{xg11b}.} $G_{SI}$ is shown for an example covalent modification cycle in Figure \ref{fig:gSI}. $G_{SI}$ is strongly connected as long as $G_E$ and $G_F$ are strongly connected, as we assume. 

To prove \eqref{Force-from-CE}, we begin with equation \eqref{SIgraph} from above, which says:
\begin{equation}\label{xu_ratio}
\frac{[S_1]}{[S_0]}=\frac{c_0^E[E]+c_0^F[F]}{c_1^E[E]+c_1^F[F]} = \frac{c_0^E([E]/[F])+c_0^F}{c_1^E([E]/[F])+c_1^F} 
\end{equation}
An alternative expression for this ratio can be obtained from the MTT applied to the graph $G_{SI}$. Since $G_{SI}$ is assumed to be reversible, it is helpful---to organize the expressions that arise from the MTT---to consider undirected spanning trees in the corresponding undirected graph $G^u_{SI}$, which is $G_{SI}$, but with the directions of the edges ignored. Given any vertex $i$, any spanning tree $T$ of $G^u_{SI}$ can be uniquely ``lifted'' to a directed spanning tree $T_i$ rooted at $i$, and all rooted spanning trees arise in this way.\footnote{In more detail, there is a bijection $\Phi_{i,j} : \Theta_i(G_{SI}) \to \Theta_j(G_{SI})$ between spanning trees rooted at any two vertices $i$ and $j$ obtained by reversing  all the edges along the unique directed path from $j$ to $i$ \cite{wcg18}.} See Figure \ref{fig:lifting} for an illustration of this construction.

\begin{figure}
	\begin{center}
		\includegraphics[scale=0.5]{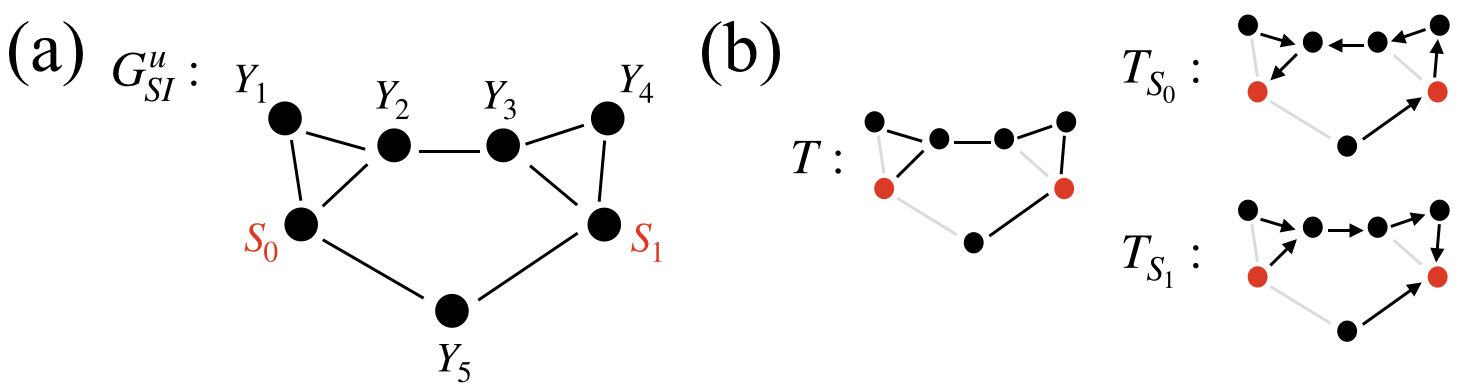}
	\end{center}
	\caption{{\bf Lifting the undirected spanning trees of $G_{SI}^u$.} (a) The graph $G^u_{SI}$ associated with the covalent modification cycle of Figure \ref{fig:gSI}, together with (b) an example spanning tree $T$ of $G^u_{SI}$ and the associated directed spanning trees $T_{S_0}$ and $T_{S_1}$ rooted at $S_0$ and $S_1$, respectively.}
	\label{fig:lifting}
\end{figure}

Applied to $G_{SI}$, the MTT yields:
\begin{equation}
\frac{[S_1]}{[S_0]}=\frac{\sum_T w(T_{S_1})}{\sum_T w(T_{S_0})}
\end{equation}
where the sum is over the spanning trees of $G^u_{SI}$, and $T_{S_1}$, $T_{S_0}$ are the directed spanning trees rooted at $S_1$ and $S_0$, respectively.

Every spanning tree of $G^u_{SI}$ contains a unique path between $S_0$ and $S_1$---if there were multiple paths, it would fail to be a tree, and if there were no path, the tree would fail to include every vertex (it would fail to be ``spanning"). Let $U$ be the set of spanning trees where the path between $S_0$ and $S_1$ involves intermediates containing $E$. Let $V$ be the set of spanning trees where the path between $S_0$ and $S_1$ involves intermediates containing $F$. Every spanning tree of $G^u_{SI}$ lies in either $U$ or $V$, never both.

Furthermore, if $T \in U$, the weights $w(T_{S_0})$, $w(T_{S_1})$ must be linear in $[E]$, because the directed trees must leave $S_1$ (resp. $S_0$) by exactly one edge, and it must be an edge carrying $[E]$ on its label since $T$ lies in $U$. Similarly, if $T$ lies in $V$, the weights $w(T_{S_0})$, $w(T_{S_1})$ must be linear in $[F]$.

It follows that we can write
\begin{equation}\label{alternate}
\frac{[S_1]}{[S_0]}=\frac{\sum_{T \in U} w(T_{S_1})+\sum_{T \in V} w(T_{S_1})}{\sum_{T \in U} w(T_{S_0})+\sum_{T \in V} w(T_{S_0})} = \frac{A [E] + B [F]}{C [E] + D [F]} = \frac{A ([E]/[F]) + B}{C ([E]/[F]) + D}
\end{equation}
where $A$, $B$, $C$, and $D$ are positive constants that do not depend on $[E]$ or $[F]$. We also have 
\begin{equation}
\frac{A D}{B C} = \frac{A [E] D [F]}{C [E] B [F] } = \frac{\sum_{T \in U} w(T_{S_1})\sum_{T \in V} w(T_{S_0})}{\sum_{T \in U} w(T_{S_0}) \sum_{T \in V} w(T_{S_1})}
\end{equation}

To proceed, we need the following claim.

\begin{claim}
	Suppose $a_1, b_1, c_1, d_1, a_2, b_2, c_2, d_2$ are nonzero real numbers such that
	\begin{equation}\label{mobiuseq}
	\frac{a_1 x + b_1}{c_1 x + d_1} = \frac{a_2 x + b_2}{c_2 x + d_2}
	\end{equation} 
	for all $x$. Then
	\begin{equation}
	\frac{a_1 d_1}{b_1 c_1} = \frac{a_2 d_2}{b_2 c_2}.
	\end{equation} 
\end{claim}
\begin{proof}
	\eqref{mobiuseq} implies 
	\begin{equation}
	(a_1 x + b_1)(c_2 x + d_2) = (a_2 x + b_2)(c_1 x + d_1).
	\end{equation}
	Expanding yields $(a_1 c_2 - a_2 c_1) x^2 + (a_1 d_2 - a_2 d_1 +b_1 c_2 - b_2 c_1) x + (b_1 d_2 - b_2 d_1) = 0$. Since this holds for all $x$, we find that $b_1 d_2 = b_2 d_1$ and $a_1 c_2 = a_2 c_1$. This means
	\begin{equation}
	1 = \frac{b_1 d_2}{b_2 d_1} = \frac{a_1 c_2}{a_2 c_1}
	\end{equation}
	from which the desired result follows by rearrangement.
\end{proof}

Equating \eqref{xu_ratio} and \eqref{alternate} and applying the claim yields
\begin{equation}
\frac{c_0^E c_1^F}{c_1^E c_0^F} = \frac{\sum_{T \in U} w(T_{S_1})\sum_{T \in V} w(T_{S_0})}{\sum_{T \in U} w(T_{S_0}) \sum_{T \in V} w(T_{S_1})}.
\end{equation}
Now let $T$ be an tree of $G_{SI}$ and suppose that $P$ is the unique \textit{directed} path from $S_0$ to $S_1$ in $T_{S_1}$. Let $P^*$ denote the reverse directed path from $S_1$ to $S_0$. Then it is easy to see that $w(T_{S_1})/w(T_{S_0}) = w(P)/w(P^*)$. By assumption that the thermodynamic force is the same about any cycle in our reaction network, \textit{this quantity depends only on whether $T$ lies in $U$ or $V$.}

It follows that, 
\begin{equation}
\frac{c_0^E c_1^F}{c_1^E c_0^F} = \frac{w(P)w(Q)}{w(P^*)w(Q^*)}, 
\end{equation}
where $P$ is any directed path from $S_0$ to $S_1$ through the intermediates containing $E$, $Q$ is any directed path from $S_1$ to $S_0$ through the intermediates containing $F$, and $P^*, Q^*$ are the respective reverses of those paths. 

Accordingly,
\begin{equation}
\log\left(\frac{c_0^E c_1^F}{c_1^E c_0^F}\right) = \log\left(\frac{w(P)w(Q)}{w(P^*)w(Q^*)}\right) = \log\left(\frac{w(C)}{w(C^*)}\right)
\end{equation}
where $C$ is the directed cycle in $G_{SI}$ formed by the concatenation of $P$ and $Q$. But the term $\log\left(\frac{w(C)}{w(C^*)}\right)$ is exactly a log ratio of rates about a cycle, as we considered in our physical discussion above. It is therefore equal to the thermodynamic force, or chemical potential difference $\Delta \mu / k_\mathrm{B}T$, holding the system out of equilibrium. This proves \eqref{Force-from-CE}.

\subsection*{The dynamic range}

The groundwork we have laid now allows us to establish bounds---in terms of $\Delta \mu / k_\mathrm{B} T$---on the characteristics of a switch based on any covalent modification cycle, taking the ``output" variable of the switch to be the logarithm of steady-state ratio $[S_1] / [S_0]$ and ``input" variables to be the conserved enzyme totals.

The simplest characteristic of a switch is the difference between the largest and smallest values its output variable can assume---the \emph{dynamic range}. From \eqref{SIgraph}, we have:
\begin{equation}
\log \left( \frac{c_0^F}{c_1^F}\right) < \log\left(\frac{[S_1]}{[S_0]}\right) < \log \left( \frac{c_0^E}{c_1^E}\right)
\label{dynamicrange}
\end{equation} 
and these limits can be approached---the larger when $F_\mathrm{tot} \to 0$ and the smaller when $E_\mathrm{tot} \to 0$. The dynamic range of the switch is then the difference of these extremes, which by  \eqref{Force-from-CE}, is simply equal to $\Delta \mu / k_\mathrm{B} T$. In the special case where each enzyme forms only one intermediate enzyme-substrate complex, this thermodynamic bound on the dynamic range was found by Qian \cite{Qian_2003}.

\subsection*{The high substrate regime}

We now turn to the \textit{sensitivity} of a switch---how sharply can the output respond to a small change in the input? Our main result is that, in the limit where $S_\mathrm{tot}$ is very large compared to the enzyme totals $E_\mathrm{tot}$ and $F_\mathrm{tot}$ and the generalized Michaelis-Menten constants, the sensitivity is bounded by a simple function of the switch parameters. To be precise, the ratio $[S_1]/[S_0]$ has a well-defined limit as $S_\mathrm{tot} \to \infty$, and it is this limiting value, which depends on the rate constants and conserved enzyme totals, whose sensitivity we bound: 
\begin{equation}\label{highStot-pre}
\left|\frac{\partial}{\partial \log E_\mathrm{tot}} \log \left(\lim_{S_\mathrm{tot} \to \infty}  \frac{[S_1]}{[S_0]}\right)\right| \leq \frac{\sqrt{c_0^E c_1^F/c_1^E c_0^F}-1}{2}.
\end{equation}
By \eqref{Force-from-CE}, the right hand side can expressed in terms of the thermodynamic force alone:
\begin{equation} \label{highStot}
\left|\frac{\partial}{\partial \log E_\mathrm{tot}} \log \left(\lim_{S_\mathrm{tot} \to \infty}  \frac{[S_1]}{[S_0]}\right)\right|  \leq \frac{\exp\left(\Delta \mu/2 k_\mathrm{B} T\right)-1}{2}.
\end{equation}

\eqref{highStot} is our main result. We note that in the limit of thermodynamic irreversibility $\Delta \mu/k_\mathrm{B} T \to \infty$, the right hand side of \eqref{highStot} diverges, consistent with the finding of ``unlimited'' sensitivity by Goldbeter and Koshland, as discussed in the Introduction.  

To prove the bound \eqref{highStot-pre}, our starting point is \eqref{Invariant-2}.  In prior work \cite{dcg12}, it was shown that \eqref{Invariant-2} implies the existence of the limit
\[
\sigma = \lim_{S_\mathrm{tot} \to \infty} [S_1]/[S_0],
\] 
and that the limiting value $\sigma$ is the unique positive solution of the quadratic equation (Eq.~(18), \cite{dcg12}):
\begin{linenomath}
	\begin{align}
	\frac{E_{\textmd{tot}}}{F_{\textmd{tot}}}\Big(\frac{1}{K_{0}^F} + \frac{\sigma}{K_{1}^F}\Big)(c_{0}^E - c_{1}^E \sigma) &= \Big(\frac{1}{K_{0}^E} + \frac{\sigma}{K_{1}^E}\Big)(c_{1}^F\sigma - c_{0}^F).
	\label{limiting_invariant}
	\end{align}
\end{linenomath}
This equation can be rearranged to express $\frac{E_{\textmd{tot}}}{F_{\textmd{tot}}}$ as a rational function of $\sigma$,
\begin{linenomath}
	\begin{align}
	\frac{E_{\textmd{tot}}}{F_{\textmd{tot}}} &= \frac{\Big(\frac{1}{K_{0}^E} + \frac{\sigma}{K_{1}^E}\Big)(c_{1}^F\sigma - c_{0}^F)}{\Big(\frac{1}{K_{0}^F} + \frac{\sigma}{K_{1}^F}\Big)(c_{0}^E - c_{1}^E \sigma)},
	\end{align}
\end{linenomath}
with numerator and denominator both positive for any value of $\sigma$ within the bounds on its value set by our dynamic range result \eqref{dynamicrange}.  This means $E_\mathrm{tot}/F_\mathrm{tot}$ can be viewed as a continuously differentiable function of $\sigma$.

Rewriting \eqref{limiting_invariant} in terms of variables $x = \log\left(\frac{E_{\textmd{tot}}}{F_{\textmd{tot}}}\right)$ and $y = \log{\sigma}$, we have:
\begin{equation}
x = \log{(K_1^E + K_0^E e^{y})} + \log{(c_1^F e^{y} - c_0^F)} - \log{(K_1^F + K_0^F e^{y})} - \log{(c_0^E - c_1^E e^{y})} + \log{\Bigg(\frac{K_0^F K_1^F}{K_0^E K_1^E}\Bigg)}.
\end{equation}

Our goal is to bound the derivative $d y / d x$. We will do this by studying the derivative of the inverse, which is:
\begin{equation}
\frac{dx}{dy} = \left( \frac{e^{y}}{e^{y} - \frac{c_0^F}{c_1^F}} + \frac{e^{y}}{\frac{c_0^E}{c_1^E} - e^{y}} \right) +  \left( \frac{e^{y}}{\frac{K_1^E}{K_0^E} + e^{y}} - \frac{e^{y}}{\frac{K_1^F}{K_0^F}+ e^{y}} \right).
\label{inverse_derivative}
\end{equation}

The second term in brackets is the difference of two positive quantities no larger than one, so it can be no smaller than $-1$. It approaches $-1$ when $\frac{K_1^E}{K_0^E} \rightarrow \infty$ and $\frac{K_1^F}{K_0^F} \rightarrow 0$. 

The first term in brackets is minimized, for fixed values of the tgCEs, when $y= \log{\sqrt{\frac{c_0^F c_0^E}{c_1^E c_1^F}}}$, when it takes the value $\left(\sqrt{\frac{c_1^F c_0^E}{c_0^F c_1^E}}+1\right)/\left(\sqrt{\frac{c_1^F c_0^E}{c_0^F c_1^E}}-1\right)$.

Therefore, 
\begin{equation}
\frac{dx}{dy}  \geq \frac{\sqrt{\frac{c_1^F c_0^E}{c_0^F c_1^E}}+1}{\sqrt{\frac{c_1^F c_0^E}{c_0^F c_1^E}} - 1} - 1 > 0.
\end{equation}
Since the function relating $x$ and $y$ is continuously differentiable and this derivative is never zero, this implies
\begin{equation}
\frac{dy}{dx}  \leq \left(\frac{\sqrt{\frac{c_1^F c_0^E}{c_0^F c_1^E}}+1}{\sqrt{\frac{c_1^F c_0^E}{c_0^F c_1^E}} - 1} - 1\right)^{-1} = \frac{\sqrt{c_0^E c_1^F/c_1^E c_0^F}-1}{2}.
\end{equation}
which establishes \eqref{highStot-pre}. Finally, applying \eqref{Force-from-CE} yields \eqref{highStot}. See Figure \ref{fig:bounds}(a) for a numerical illustration of this result. 

In prior work, Qian \cite{Qian_2003}---studying a reversible covalent-modification switch with a single intermediate for each enzyme---gave an asymptotic formula for the derivative $\partial ([S_1]/S_\mathrm{tot}) / \partial \log E_\mathrm{tot}$, in the high substrate regime and at a point where $[S_1]/S_\mathrm{tot} = 1/2$. Starting from our \eqref{inverse_derivative}, and taking $y = 0$, we can develop a similar expression:
\begin{equation}
\frac{\partial}{\partial \log E_\mathrm{tot}} \log \left(\lim_{S_\mathrm{tot} \to \infty}  \frac{[S_1]}{S_\mathrm{tot}}\right) =\frac{1}{4} \left[\left( \frac{1}{1 - \frac{c_0^F}{c_1^F}} + \frac{1}{\frac{c_0^E}{c_1^E} - 1} \right) + \left( \frac{1}{\frac{K_1^E}{K_0^E} + 1} - \frac{1}{\frac{K_1^F}{K_0^F}+ 1} \right) \right]^{-1},
\end{equation}
when $[S_1]/S_\mathrm{tot} = 1/2$. This result may be compared to Eq.~(17) of Qian \cite{Qian_2003}, but it is not equivalent to it.  Our equation reduces, under the multiple limits simultaneously taken by Qian, to the thermodynamic parts of his expression. However, Qian's equation additionally includes terms proportional to $1 / S_\mathrm{tot}$, which cannot appear in our expression because we have already taken the limit $S_\mathrm{tot} \to \infty$. 

\subsection*{The low substrate regime}
The sensitivity bound \eqref{highStot} for the high substrate regime can be viewed as a companion to a corresponding result in the low substrate regime \cite{Owen_Horowitz_2020}:
\begin{equation} \label{lowStot}
\left| \frac{\partial}{\partial \log E_\mathrm{tot}} \log \left(\lim_{S_\mathrm{tot}\to 0} \frac{[S_1]}{[S_0]}\right)\right| \leq \tanh\left(\Delta \mu/4 k_\mathrm{B} T\right).
\end{equation}
Note that the right hand side is less than $1$: there can be no ultrasensitivity in the low substrate regime. This result is a manifestation of a recently identified \cite{Owen_Horowitz_2020} universal thermodynamic bound on the response of nonequilibrium systems to perturbations. It also follows directly from the approach described here, as we now show.  

In the limit where $S_\mathrm{tot}$ is very small, we have $[E] \to E_\mathrm{tot}$, $[F] \to F_\mathrm{tot}$, which together with \eqref{xu_ratio} implies
\begin{equation}
\lim_{S_\mathrm{tot} \to 0} \frac{[S_1]}{[S_0]} = \frac{c_0^E(E_\mathrm{tot}/F_\mathrm{tot})+c_0^F}{c_1^E(E_\mathrm{tot}/F_\mathrm{tot})+c_1^F} .
\end{equation}
Let $\phi$ denote this limiting quantity and set $z = E_\mathrm{tot} / F_\mathrm{tot}$. Taking the derivative, we find that
\[
\frac{\partial}{\partial \log E_\mathrm{tot}} \log (\phi) = \frac{c_0^E c_1^F z- c_0^F c_1^E z}{(c_0^F + c_0^E z)(c_1^F+c_1^E z)} = \frac{\left(\sqrt{c_0^E c_1^F z}- \sqrt{c_0^F c_1^E z}\right)\left(\sqrt{c_0^E c_1^F z}+ \sqrt{c_0^F c_1^E z}\right)}{(c_0^F + c_0^E z)(c_1^F+c_1^E z)}. 
\]
Using the inequality of the arithmetic and geometric means, 
\[
(c_0^F + c_0^E z)(c_1^F+c_1^E z) \geq \left(\sqrt{c_0^E c_1^F z}+ \sqrt{c_0^F c_1^E z}\right)^2.
\]
Recalling that $\tanh(x) = (e^{2x} - 1)/(e^{2x} + 1)$, we see that, 
\begin{equation}
\left|\frac{\partial}{\partial \log E_\mathrm{tot}} \log (\phi) \right| \leq \frac{\sqrt{c_0^E c_1^F z}- \sqrt{c_0^F c_1^E z}}{\sqrt{c_0^E c_1^F z}+ \sqrt{c_0^F c_1^E z}} = \tanh\left(\frac{1}{4}\log\left(\frac{c_0^E c_1^F}{c_0^F c_1^E}\right)\right).
\end{equation}
Finally, applying relation \eqref{Force-from-CE} yields \eqref{lowStot}. See Figure \ref{fig:bounds}(b) for a comparison of this bound to that for the high substrate regime. The bound on the sensitivity in the high substrate regime is larger than the bound in the low substrate regime for all nonzero values of $\Delta \mu/ k_\mathrm{B} T$.  For small $\Delta \mu/ k_\mathrm{B} T$, the two bounds are equal to first order. 

\begin{figure}
	\begin{center}
		\includegraphics[scale=0.6]{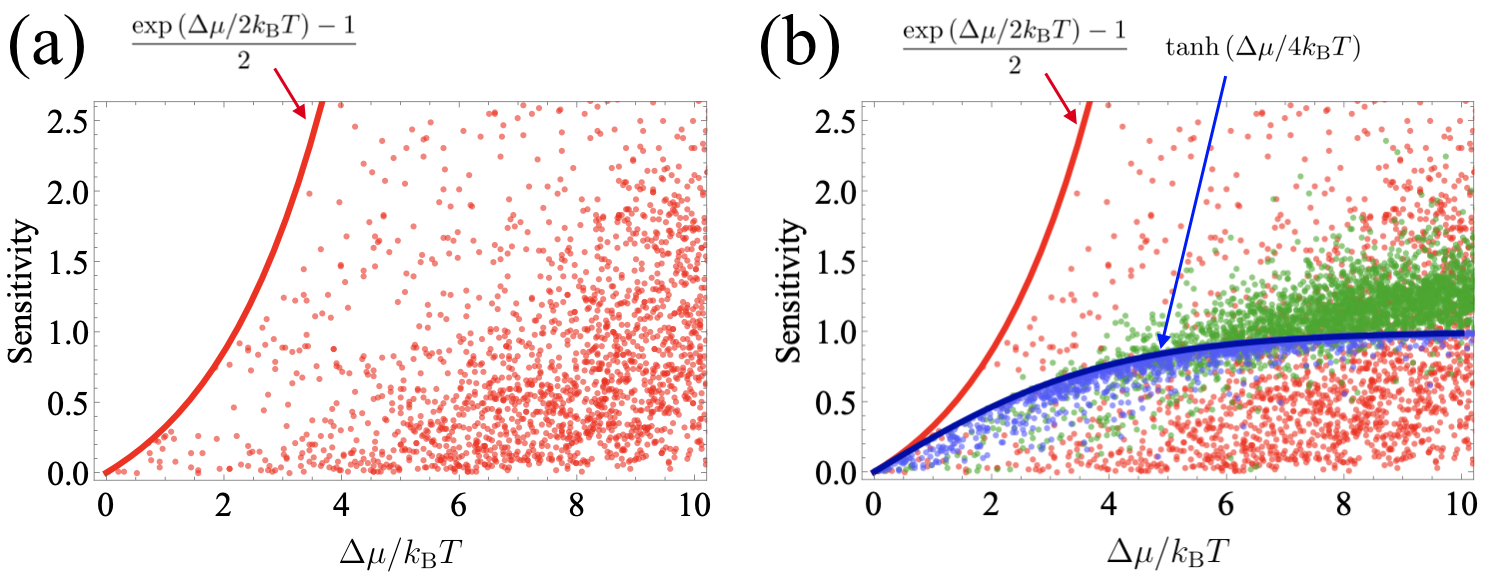}
	\end{center}
	\caption{{\bf Thermodynamic force bounds sensitivity.} Numerical illustration of our bounds by random sampling of rate constants for a covalent modification cycle in which both enzymes have single intermediates (see Methods for details). (a) Points show values of the thermodynamic force and sensitivity achieved when $S_\mathrm{tot} = 20000$ (red). In (b), the same is plotted together with $S_\mathrm{tot} = 2$ (green), and $S_\mathrm{tot} = 0.01$ (blue). In all cases $E_\mathrm{tot} = F_\mathrm{tot} = 1$. These numerical results are compared to the bounds \eqref{highStot} (solid red curve) and \eqref{lowStot} (solid blue curve).}
	\label{fig:bounds}
\end{figure}

\subsection*{Saturating our bounds}

Remarkably, at any fixed force, there are covalent modification cycles that, with the same kinetic parameters, can saturate either bound \eqref{lowStot} or \eqref{highStot}, as the substrate concentration is varied. Maximal sensitivity, in both $S_\mathrm{tot}$ regimes, can be achieved when the enzymes act as ``mirrors" of each other (i.e.~the network is invariant under the exchange $E \leftrightarrow F$ and $S_0 \leftrightarrow S_1$). 

For example, consider a covalent modification cycle with $c_0^E = c_1^F = \exp\left(\Delta \mu / 2 k_\mathrm{B}T\right)$, $c_1^E  = c_0^F = 1$, and the ``forward" generalized Michaelis-Menten constants  $K_0^E = K_1^F \equiv K_M$ very small compared to the ``reverse" constants $K_1^E = K_0^F \equiv K_R$. For such a covalent modification cycle, \eqref{inverse_derivative} gives the maximal sensitivity in the high $S_\mathrm{tot}$ regime explicitly as:
\begin{equation}
\frac{\exp\left(\Delta \mu / 2 k_\mathrm{B}T\right)-1}{2} \times \frac{1+K_M/K_R}{1+ \exp\left(\Delta \mu / 2 k_\mathrm{B}T\right) K_M/K_R},
\end{equation} 
saturating our bound \eqref{highStot} when $K_R  \gg  K_M$. For the same parameters, in the low $S_\mathrm{tot}$ regime we have
\begin{equation}
\frac{[S_1]}{[S_0]} \approx \frac{c_0^E (E_\mathrm{tot}/F_\mathrm{tot})+ c_0^F}{c_1^E (E_\mathrm{tot}/F_\mathrm{tot}) + c_1^F} = \frac{\exp\left(\Delta \mu / 2 k_\mathrm{B}T\right)(E_\mathrm{tot}/F_\mathrm{tot})+ 1}{(E_\mathrm{tot}/F_\mathrm{tot}) + \exp\left(\Delta \mu / 2 k_\mathrm{B}T\right)},
\end{equation}
leading to a maximal sensitivity of $\tanh\left(\Delta \mu / 4 k_\mathrm{B}T\right)$, saturating the low $S_\mathrm{tot}$ bound \eqref{lowStot}. These results are illustrated in Figure \ref{fig:saturation}.

\begin{figure}
	\begin{center}
		\includegraphics[scale=0.45]{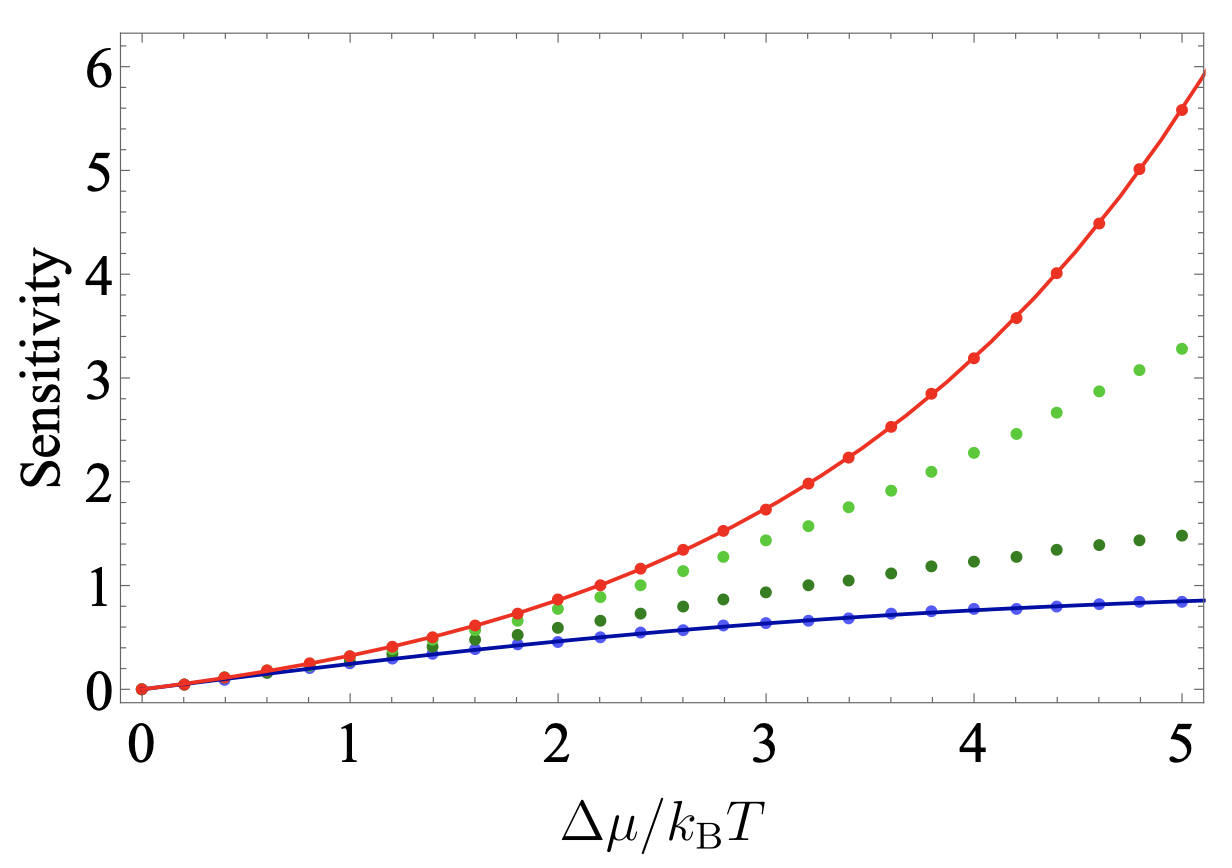}
	\end{center}
	\caption{{\bf The same parameter choice can saturate the low and high $S_\mathrm{tot}$ bounds.} Sensitivity bounds in the high (solid red line) and low (solid blue line) $S_\mathrm{tot}$ regimes compared to the sensitivity (colored points) evaluated at $E_\mathrm{tot} = F_\mathrm{tot} = 1$ for a covalent modification cycle with $K_0^E = K_1^F = 10^{-4}$, $K_1^E = K_0^F = 10^3$, $c_0^E = c_1^F = \exp\left(\Delta \mu / 2 k_\mathrm{B}T\right)$, $c_1^E  = c_0^F = 1$. Points corresponding to $S_\mathrm{tot}$ values of $100$, $1.75$, $1$, and $0.01$ are colored red, light green, dark green, and blue, respectively.}
	\label{fig:saturation}
\end{figure}

\section*{Discussion}

The bounds we have presented provide a quantitative picture of how general covalent modification cycles are constrained by thermodynamics. We have framed our results in term of the chemical potential $\Delta \mu$, which is the nonequilibrium driving force of the system, and is a natural way to quantify the energetic requirements of the switch. To see this, it is important to note that there can be no \emph{direct} trade-off between the \emph{rate} of (free) energy expenditure---the power---and steady-state properties of the system like its sensitivity or dynamic range. This is because the power can be scaled arbitrarily by scaling all reaction rates by the same constant, or, equivalently, by rescaling time, whereas this operation leaves the steady state completely unchanged. To compare it to steady-state quantities, power must be divided by some other rate relevant to the system. The quantity $\Delta \mu$ is equal to the energy consumed per cycle completion, i.e., it is the power divided by the rate of cycle completion.    

Our results highlight the fact that under common circumstances, thermodynamics does not tightly constrain the sensitivity of biochemical switches. For example, in a phosphorylation-dephosphorylation cycle driven by ATP hydrolysis, typical physiological values of $\Delta \mu \approx 20 - 30 \, k_\mathrm{B} T$ are deep in the saturated regime of the low $S_\mathrm{tot}$ bound, and in the high $S_\mathrm{tot}$ case, lead to a maximum possible sensitivity of $\sim 10^4 \text{ to } 10^6$, far in excess what is needed to account for typical Hill coefficients measured in ultrasensitive systems \cite{Ferrell2014}. The sensitivity of such strongly driven switches may be constrained by other factors, such as kinetics or the abundance of substrate. 

Our bounds also show that, from the perspective of making a good switch based on covalent modification, enzymological complexity provides no benefit. No matter the number of enzymatic intermediates or how elaborate their reactions, the same thermodynamic bounds on sensitivity hold. This is in contrast to numerous examples in biophysics where having more states or ``steps" provides some advantage. For example, in kinetic proofreading, having more proofreading steps allows for a degree of discrimination that can never be achieved with fewer, even as $\Delta \mu \to \infty$ \cite{murugan2014discriminatory, ouldridge2017thermodynamics, Owen_Horowitz_2020}. As another example, the maximum possible coherence of biochemical oscillations is conjectured to depend strongly on the number of states available, together with the strength of nonequilibrium driving \cite{barato2017coherence}. 

A number of basic questions remain. First, what is the bound on sensitivity in terms of $\Delta \mu$ \textit{and} $S_\mathrm{tot}$? We have only studied the limiting cases in which substrate is very scarce or abundant. It is natural to conjecture that the maximum possible sensitivity is increasing in $S_\mathrm{tot}$, and that therefore our high substrate bound \eqref{highStot} in fact holds for all $S_\mathrm{tot}$.

Second, in this paper we have focused on steady-state behavior. It would be very interesting to also understand the constraints on the transient behavior. For example, in vision, the exceptional amplification that enables rod cells to respond one or a few photons involves a transient response of a modification cycle involving rhodopsin and transducin, driven by GTP hydrolysis \cite{arshavsky2014current, baylor1996photons}.

Finally, we note that even simpler properties of general covalent modification cycles remain incompletely understood. For example, to our knowledge, it remains an open problem to prove the monotonicity of the steady-state ratio $[S_1]/[S_0]$ as a function of the enzyme totals. Such matters may at first seem only of mathematical interest, but we think understanding them carefully could bear fruit---especially in the study of systems, such as signaling cascades, in which covalent modification cycles appear as parts \cite{barabanschikov2019monostationarity}.

\section*{Methods}

To generate Figure \ref{fig:bounds}, the sensitivity was evaluated numerically for the simplest reversible covalent modification cycle
\[
\begin{array}{c}
	E + S_0  \overset{a}{\underset{b}\rightleftarrows} 
	 ES  \overset{c}{\underset{d}\rightleftarrows}  E + S_1 \\
	F + S_1 \overset{s}{\underset{r}\rightleftarrows} FS  \overset{q}{\underset{k}\rightleftarrows} F + S_0
\end{array}
\]
for random choices of the rate constants $a$, $b$, $c$, $d$, $s$, $r$, $q$, and $k$. Explicitly, for the $S_\mathrm{tot} = 20000$ (red) and $S_\mathrm{tot} = 2$ (green) points, all rate constants were drawn uniformly from the interval $(0, 10)$ except for $d$ and $k$ which were drawn uniformly from $(0, 0.1)$. For the $S_\mathrm{tot} = 0.01$ (blue) points, all rate constants were drawn uniformly from the interval $(0, 10)$ except for $d$ and $k$ which were drawn uniformly from $(0, 0.5)$. The polynomial equations defining the steady state where solved numerically (using {\tt NSolve} in \emph{Wolfram Mathematica}), and the derivative defining the sensitivity at $E_\mathrm{tot} = F_\mathrm{tot} = 1$ was estimated by taking a finite difference. For this covalent modification cycle, $\Delta \mu / k_\mathrm{B}T = \log\left(\frac{acsq}{bdrk}\right)$.

To generate Figure \ref{fig:saturation}, \eqref{Invariant-1} and \eqref{Invariant-2} were solved numerically, with $K_0^E = K_1^F = 10^{-4}$, $K_1^E = K_0^F = 10^3$, $c_0^E = c_1^F = \exp\left(\Delta \mu / 2 k_\mathrm{B}T\right)$, $c_1^E  = c_0^F = 1$, and varying values of $S_\mathrm{tot}$. The sensitivity at $E_\mathrm{tot} = F_\mathrm{tot} = 1$ was estimated by taking a finite difference.  

\section*{Acknowledgments}

JWB and JG were supported by grant 1462629 from the US National Science Foundation.

\bibliographystyle{plain}

\end{document}